%% file: main.tex
\documentclass[3p]{elsarticle}

\usepackage{graphicx}      

\usepackage{geometry}
\makeatletter
\def\ps@pprintTitle{%
 \let\@oddhead\@empty
 \let\@evenhead\@empty
 \def\@oddfoot{\centerline{\thepage}}%
 \let\@evenfoot\@oddfoot}
\makeatother

\input{settings.tex}
\def\withproofs{1}
\def\arxivpaper{1}
\begin{document}
 \title{\bf On Excessive Transverse Coordinates for Orbital Stabilization of Periodic Motions\tnoteref{t1}} 
\tnotetext[t1]{This work has been  supported by the Research Council of Norway,  grant number 262363.}

\author{Christian Fredrik Sætre   \&
Anton Shiriaev} 

\address{Department of Engineering Cybernetics, NTNU, Trondheim, Norway. \newline
 \{christian.f.satre,anton.shiriaev\}@ ntnu.no}

\begin{abstract}                
This paper explores transverse coordinates for the purpose of orbitally stabilizing periodic motions of nonlinear, control-affine dynamical systems. It is shown that the dynamics of any (minimal or excessive) set of transverse coordinates, which are defined in terms of a particular parameterization of the motion and a strictly state-dependent projection operator recovering the parameterizing variable, admits a (transverse) linearization along the target motion, with explicit expressions stated. Special focus is then placed on a generic excessive set of orthogonal coordinates, revealing a certain limitation of the ``excessive" transverse linearization for the purpose of control design. To overcome this limitation, a linear comparison system is introduced and conditions are stated for when the asymptotic stability of its origin corresponds to the asymptotic stability of the origin of linearized transverse dynamics. This allows for the construction of feedback controllers utilizing this comparison system which, when implemented on the dynamical system, renders the desired motion asymptotically stable in the orbital sense.
\end{abstract}


\maketitle
\input{introduction.tex}
\input{PreliminariesAndKeyIdea.tex}

\input{EquivalenceBetweenTransverseCoorindates.tex}
\input{Generic_set_of_Excessive_Ort_Coors.tex}
\input{LimitationsOfTVL.tex}

\input{ComparisonSys.tex}
\input{exampleBanHauser.tex}
\input{conclusion.tex}

\bibliographystyle{elsarticle-harv}
\bibliography{references.bib}             

\if\withproofs1
\appendix
\input{appendixTVL.tex}
\input{appendix.tex}
\fi
\end{document}

%% file: settings.tex
\usepackage{amsmath}
\usepackage{amssymb}
\usepackage{amsthm}
\usepackage[utf8]{inputenc}
\usepackage[dvipsnames]{xcolor}
\usepackage{tikz}
\usepackage{siunitx}

\colorlet{lightGray}{black!10!white}
\colorlet{lightlightGray}{black!5!white}

\usetikzlibrary{backgrounds}

\usetikzlibrary{calc,patterns,angles,quotes,decorations.pathmorphing}
\usetikzlibrary{shadows.blur}
\usetikzlibrary{shapes.symbols,positioning}

\newtheorem{thm}{\bf Theorem}
\newtheorem{lem}{\bf Lemma}
\newtheorem{prop}{\bf Proposition}
\newtheorem{cor}{\bf Corollary}
\theoremstyle{definition}
\newtheorem{rem}{\bf Remark}
\newtheorem{defn}{\bf Definition}

\newcommand{\Ri}[1]{\mathbb{R}^{#1}}
\newcommand{\cont}[1]{\mathcal{C}^{#1}}
\newcommand{\I}[1]{{{ I}}_{#1}}
\newcommand{\0}[1]{{{\bf 0}}_{#1}}

\DeclareMathOperator*{\argmin}{arg\,min}

\newcommand{\inv}{^{-1}}
\newcommand{\pinv}{^{\dagger}}
\newcommand{\transp}{^{\mathsf{T}}}
\newcommand{\rank}[1]{{\text{rank}~{ #1}}}
\newcommand{\im}[1]{{\text{im}~{ #1}}}
\newcommand{\vspan}[1]{{\text{span}~{ #1}}}

\newcommand{\mg}{s}
\newcommand{\xs}{x_s}
\newcommand{\tvc}{z_{\perp}}
\newcommand{\Dtvc}{\dot{z}_ {\perp}}
\newcommand{\state}{{{x}}}
\newcommand{\Dstate}{{\dot{{x}}}}
\newcommand{\ac}{{{u}}}

\newcommand{\nvel}{\rho}
\newcommand{\sspace}{\mathcal{S}}
\newcommand{\Pspace}{\mathcal{X}}
\newcommand{\prj}{p}
\newcommand{\DP}{D\prj}
\newcommand{\DDP}{D^2\prj}
\newcommand{\DPs}{\Gamma}
\newcommand{\DDPs}{D^2\prj}
\newcommand{\Oms}{\Omega}
\newcommand{\Pis}{\Pi}
\newcommand{\Pisin}{\Pi^{\dagger}}
\newcommand{\mtvc}{y_\perp}
\newcommand{\Dmtvc}{\dot{y}_\perp}

\newcommand{\pperp}{\varphi_{\perp}}
\newcommand{\Pperp}{\Phi_{\perp}}

\newcommand{\Pperpinv}{\Phi_{\perp}\pinv}

%% file: introduction.tex
\section{INTRODUCTION}
We consider the task of designing orbitally stabilizing feedback for periodic solutions of nonlinear, control-affine dynamical systems, defined by
\begin{equation}\label{eq:DynSYs}
    \Dstate=f(\state)+g(\state)\ac, \quad \state\in\Ri{n},  \quad \ac\in\Ri{m}.
\end{equation}
Here the notion of \emph{asymptotic orbital (Poincaré) stability} simply means the asymptotic  convergence to the   periodic \emph{orbit} (i.e. the set of all the states along the solution) and not to a specific point-in-time along a trajectory (see e.g. \cite{leonov2008strange}). 
In this regard,  we  recall  the following.
\begin{thm}[Andronov--Vitt]
A nontrivial, $T$-periodic solution $\state_*(t)=\state_*(t+T)$ of a smooth  dynamical  system $ \Dstate=F(\state)$ on $\Ri{n}$ is asymptotically orbitally stable if the  first approximation, $\delta \Dstate=\frac{\partial F}{\partial \state}(\state_*(t))\delta\state$,
has one simple zero characteristic  exponent and the remaining  $(n-1)$ characteristic exponents have  strictly negative real parts.
\end{thm}

It thus follows    that the stability of a periodic orbit is  equivalent to the stability of an $(n-1)$-dimensional subsystem of the first approximation along the nominal solution.  At the same time, the Andronov--Vitt theorem also highlights a limitation of the first approximation for the purpose of  feedback design for \eqref{eq:DynSYs}
due to its non-vanishing (zero characteristic (Floquet) exponent) solution. It would  therefore clearly be beneficial to instead just    target the $(n-1)$-dimensional subsystem  directly, which it turns out  is  equivalent to only considering the dynamics \emph{transverse} to the orbit. Indeed, it is known that a periodic solution is asymptotically  stable in the orbital sense if (and only if) the dynamics {transverse} to the flow along the nominal orbit are asymptotically  stable (\cite{hauser1994converse}). 

 The design of orbitally stabilizing feedback controllers  can therefore  be boiled down to two main steps: 1) Find a (minimal) set of $(n-1)$  independent \emph{transverse coordinates} which vanish on the orbit and are non-zero away from it; and then  2) Design a  controller (by some means) which stabilizes the origin of these  coordinates. 
Here the latter step is commonly achieved by linearization of the dynamics of these coordinates along the solution, a so-called \emph{transverse linearization},  allowing for feedback design utilizing well-known linear control techniques.

While there exists constructive procedures for finding such a minimal set of coordinates for certain classes of systems  \citep{shiriaev2010transverse,banaszuk1995feedback}), finding $(n-1)$ independent coordinates can be challenging in the general case. 
The main contribution of this paper is therefore to show that one instead can utilize an  \emph{excessive} set of  { transverse coordinates}. In fact, we show that any such set (minimal or excessive) will do (see Proposition~\ref{prop:minTVCandExcTVC}). In this regard, we also provide explicit  expressions for the linearized transverse dynamics of \emph{any} (minimal or excessive) set of transverse coordinates (see  Theorem~\ref{theorem:TVLGencoords} in Sec.~\ref{sec:EquivAndArbTVL}). 

In order to provide some further insight into- and highlight a limitation of the transverse linearization for an excessive set of coordinates  (see Sec.~\ref{sec:LimTVLD}), we subsequently focus on a generic  set of easy-to-compute  orthogonal coordinates  introduced in Sec.~\ref{sec:OrtCoords}. In this regard, this paper's  second major contribution is the introduction of a linear comparison system for these coordinates, which can be used for orbitally stabilizing feedback design for   systems of the form \eqref{eq:DynSYs} (see Proposition~\ref{theorem:mainResult} in Sec.~\ref{sec:AuxSys}).
In order to illustrate the proposed scheme, we consider  a constructive  example   in Sec.~\ref{sec:ExampleBanHau}, before, lastly, we  state some concluding remarks.

%% file: PreliminariesAndKeyIdea.tex
\section{Preliminaries and key idea}\label{sec:Preliminaries}
Consider the control-affine system \eqref{eq:DynSYs}  with   $f:\Ri{n}\to\Ri{n}$ continuously differentiable   and $g(x)=[g_1(x),\dots,g_m(x)]$ with (locally) Lipschitz continuous  vector fields $g_i:\Ri{n}\to\Ri{n}$. Let $x_*(t)=x_*(t+T)$  denote a bounded,  $T$-periodic solution of the undriven system ($\ac\equiv 0$) satisfying $\|\Dot{x}_*(t)\|>0$ for all $t\ge 0$, and let 
\begin{equation*}
    \eta_*:=\{x\in\Ri{n}: \ x=x_*(t),\ t\in[0,T)\} 
\end{equation*}
denote the corresponding closed orbit.
Suppose this orbit admits a regular $\mathcal{C}^2$-parameterization, defined by
\begin{equation}\label{eq:NomOrbit}
    \xs:\sspace \to \eta_*,\quad s\mapsto \xs(s), \quad \xs(s+s_T)=\xs(s),
\end{equation}
 such that  the parameterizing   variable,   $s\in\sspace:=[s_0,s_0+s_T)$, is strictly monotonically increasing along $\eta_*$ and $\|\frac{d}{ds}\xs(s)\|=\|\xs'(s)\|>0$ for all $s\in\sspace$. 
Further suppose that a \emph{projection operator}, $\state\mapsto \prj(\state)\in\sspace$, in accordance with the  following definition is known 
for this curve.
\begin{defn}\label{def:ProjOp}
A mapping $ \prj:\Ri{n}\to\mathcal{S}$
is said to be a 
\emph{projection operator}  onto the orbit $\eta_*$ if  it is twice continuously differentiable within some  tubular neighbourhood $\Pspace\subset \Ri{n}$ of $\eta_*$ and  it is a left inverse of  the curve  \eqref{eq:NomOrbit}, that is $s=\prj( \xs(s))$ for all $ s\in\sspace$. \qed 
\end{defn}

The idea behind  such a projection operator is simply that, within some tubular neighbourhood, it allows one to project the current states  down upon the nominal orbit and consequently  define some measure of the distance to it.  For instance, consider  the set $\Lambda(\hat{s}):=\{x\in \mathcal{X}:\prj(x)=\hat{s}\}$, that is, the set of
states in a neighbourhood of $\eta_*$  mapped to some particular $\hat{s}\in\sspace$. As illustrated in Figure~\ref{fig:MPS}, it traces out a hypersurface, whose geometry is clearly dependent on the choice of $\prj(\cdot)$. This surface (manifold) of dimension $(n-1)$ is analogous to a  \emph{moving Poincaré section}  \citep{leonov2006generalization} which moves along with the trajectory and is locally  transverse to its flow.  It follows that if one can define a set of coordinates evolving upon- and spanning these sections, and then enforce,  by some control action,  strict contraction of these coordinates towards their origin (i.e. the  orbit), then  the desired trajectory must be asymptotically stable in the orbital sense.

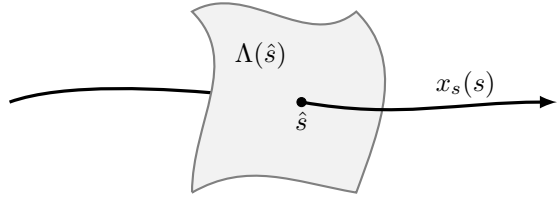
\begin{figure}
    \centering
    \input{gfx/MPS.tex}
    \caption{Illustration of the transverse surface formed by $\Lambda(\cdot)$.    }
    \label{fig:MPS}
\end{figure}

Note that this concept  is in many ways  both  similar to- and inspired by  \emph{Zhukovski stability} (see, e.g., \cite{leonov1995local,leonov2008strange}). Roughly speaking, this  notion of stability, which implies orbital stability \citep{leonov2008strange}, utilizes  parameterizations to ``align" perturbed trajectories in space   while not considering their divergence in time. Our approach, however, differs  by the fact that, whereas Zhukovski considered  reparameterizations of perturbed trajectories in terms of a ``rescaling of time", we consider a completely state-dependent projection operator as defined in Def.~\ref{def:ProjOp}. This has, for the purpose of control design,  the benefit that it allows one to define the aforementioned  state-dependent distance measure,  further allowing for the design of completely state-dependent orbitally stabilizing feedback controllers. Such a feedback, if found, then results in an   autonomous closed-loop system which admits the desired solution as an attractive  limit cycle. 

\paragraph*{\bf Notation:} 
$\|\cdot\|$  denotes the Euclidean norm.
 For a twice-continuously differentiable ($\mathcal{C}^2$-) function  $\state\mapsto h(\state)$, we denote by    $Dh(\cdot)=[\frac{\partial h}{\partial x_1}(\cdot),\dots,\frac{\partial h}{\partial x_n}(\cdot)]$ its Jacobian matrix, while if $h:\Ri{n}\to\Ri{}$, we denote by $D^2h(\cdot)$ its symmetric, $n\times n$ Hessian matrix. If $h_s(s):=h(\xs(s))$, then $h_s'(s)$ denotes the derivative $\frac{d }{ds}h_s(s)$.

%% file: gfx/MPS.tex
\begin{tikzpicture}[tight background,scale=1.2]

\if\arxivpaper0
 \coordinate (origo) at (0.2,0);
 \draw[very thick,-latex] (-3,0) to [out=20,in=-10] (origo);
\draw[gray,thick,fill=lightlightGray,shade,right color=lightGray,left color=white,middle color=lightlightGray,shading angle=135,blur shadow={shadow blur steps=100,left},shadow xshift=-0.2em,shadow yshift=-0.1em] (-1,-1) to [out=90,in=-45] (-1,1) to [out=30,in=220] (0.8,1) to [out=-45,in=70] (0.8,-1) to [out=170,in=30]  cycle ;


\if\arxivpaper0
\draw[gray,thick,fill=lightlightGray] (-1,-1) to [out=90,in=-45] (-1,1) to [out=30,in=220] (0.8,1) to [out=-45,in=70] (0.8,-1) to [out=170,in=30]  cycle ;
top color=blue!40,
      bottom color=blue!5,
      rounded corners=6pt,
      blur shadow={shadow blur steps=5}
      
\fi

\draw[very thick,-latex] (origo)  [out=-10,in=180] to (3,0);
\node at (2,0.2)  {$\state_s(s)$};
\draw[fill=black] (origo) circle(0.05) node[above left = 9pt and 2pt] {$\Lambda(\hat{s})$} node[below] {$\hat{s}$};
\fi
\if\arxivpaper1
 \coordinate (origo) at (0.2,0);
 \draw[very thick,-latex] (-3,0) to [out=20,in=-10] (origo);

\draw[gray,thick,fill=lightlightGray] (-1,-1) to [out=90,in=-45] (-1,1) to [out=30,in=220] (0.8,1) to [out=-45,in=70] (0.8,-1) to [out=170,in=30]  cycle ;

\draw[very thick,-latex] (origo)  [out=-10,in=180] to (3,0);
\node at (2,0.2)  {$\state_s(s)$};
\draw[fill=black] (origo) circle(0.05) node[above left = 9pt and 2pt] {$\Lambda(\hat{s})$} node[below] {$\hat{s}$};

\fi
\end{tikzpicture}

%% file: EquivalenceBetweenTransverseCoorindates.tex
\section{Equivalence between  coordinates and the transverse linearization}\label{sec:EquivAndArbTVL}
In  regards to the aforementioned  distance measure, consider   
\begin{equation}\label{eq:TVC}
    \tvc:=\state-\xs(\prj(\state)).
\end{equation}
 In some sense, they are the simplest measure of such a distance, but their definition is also clearly dependent on the choice of the projection operator $\prj(\cdot)$. In particular, they must  evolve upon some hypersurface such as  those formed by the set $\Lambda(\cdot)$. But $\tvc\in\Ri{n}$, and so they are an  \emph{excessive} set of coordinates upon this surface. In fact, they are not a valid change of coordinates  either, as the map $\state \mapsto \tvc$ is  evidently not a  diffeomorphism.
To see this   more clearly, consider the Jacobian matrix  $D\tvc(\state)$. 
Taking the time-derivative of  \eqref{eq:TVC}, we  obtain
 \begin{equation}\label{eq:TVD1}
     \Dtvc=D\tvc(\state)\Dstate=D\tvc(\state)f(\state)+D\tvc(\state)g(\state)u.
 \end{equation}
It follows that, sufficiently close the orbit, a  variation in the states, $\delta \state$, relates to a  variation in the coordinates \eqref{eq:TVC}  through $\Omega(s):=D\tvc(\xs(s))$:
\begin{equation}\label{eq:deltaGTVCdeltaState}
    \delta\tvc=\Oms(s)\delta \state.
\end{equation}
Similarly, by defining $\DPs(s):=D\prj(\xs(s))$, we  find that 
\begin{equation*}
    \delta s=\Gamma(s)\delta \state.
\end{equation*}
Thus for \eqref{eq:TVC} to be a valid (local) change of coordinates, the matrix function $\Oms(s)$ must necessarily be everywhere invertible. However,  as is clear by  the following statement, which is just a straightforward consequence of the relation
\begin{equation}\label{eqDpDxs}
    \DPs(s)\xs'(s)\equiv 1 \quad \forall  s\in \sspace,
\end{equation}
obtained from $s=\prj(\xs(s))$ (see Def.~\ref{def:ProjOp}),
 this can  never be the case for  solutions of the form \eqref{eq:NomOrbit}.
\begin{lem}\label{lemma:Omegas}
The matrix function
    \begin{equation}
    \Oms(s):=D\tvc(\xs(s))=\I{n}-\xs'(s)\DPs(s)
    \end{equation}
     is a projection matrix (i.e. $\Oms^2(s)=\Oms(s)$), its rank is always $(n-1)$, while  $\DPs(s):=D\prj(\xs(s))$ and $\xs'(s):=\frac{d}{ds}\xs(s)$ are its left- and right annihilators, respectively. \if\withproofs0{\qed\footnote{Proofs of all the statements are given in the extended version of this paper which is available on the arXiv: arXiv:1911.06232.}}
\fi 
\end{lem}
\if\withproofs1
Proof of this statement is given in \ref{app:proofPmegas}.
\fi

From Lemma~\ref{lemma:Omegas} it is  clear  that we  have $\Oms(s)\delta \tvc=\Oms^2(s)\delta \state=\delta \tvc$, and therefore the relation
\begin{equation*}
    \DPs(s)\delta \tvc=\DPs(s)\Oms(s)\delta \tvc\equiv 0    
\end{equation*}
must always hold. We can thus infer that, sufficiently close to the nominal orbit, the coordinates \eqref{eq:TVC} are orthogonal to the gradient of the projection operator $\prj(\cdot)$ and hence locally transverse to the nominal flow of the orbit.
Indeed,  it is important to note that the relation \eqref{eqDpDxs} does not imply that $\DPs\transp(s)$ is necessarily in the span of ${\xs'}(s)$. Rather, if $\theta(s)\in(-\frac{\pi}{2},\frac{\pi}{2})$ denotes the angle between $\DPs\transp(s)$ and ${\xs'}(s)$ in their common plane, then, as a direct consequence of the  inner product $\DPs\xs'=\|\DPs\|\|\xs'\|\cos(\theta)$,  there exists some continuously differentiable  unit vector function  $q_\perp\transp(s):\sspace\to \Ri{n}$  within  $\ker{\xs'}\transp(s)$, such that  
\begin{equation}\label{eq:GenDpsDef}
    \DPs(s)=\frac{{\xs'}\transp(s)}{\|\xs'(s)\|^2}+\tan(\theta(s))\frac{q_\perp(s)}{\|\xs'(s)\|}.
\end{equation}
Consequently, the coordinates \eqref{eq:TVC} are  in general only locally transverse to the  flow of the  orbit and not necessarily orthogonal to it. Moreover, they must be  an excessive set of transverse coordinates as $\text{rank }\Oms(s)=n-1 $. Nevertheless,  we will show shortly  that the  asymptotic stability of their origin  in fact implies the asymptotic stability of any other   valid set of transverse coordinates, and, therefore, also the asymptotic stability of the nominal orbit. 

\subsection{Equivalence between transverse coordinates}
Let us start by giving a formal definition of what we  mean when we refer to a ``valid set of transverse coordinates". In this regard, consider a $\cont{2}$-function $\mtvc:\sspace\times \Ri{n}\to \Ri{N}$,   together with a projection operator $\prj(\cdot)$. Note that we will  distinguish between the partial- and total derivative of $\mtvc$ with respect to $\state$ as follows:
    \begin{equation*}
    D\mtvc(s,x)=\frac{\partial \mtvc}{\partial \state}(s,\state)+\frac{\partial \mtvc}{\partial s}(s,\state)\DP(\state).
\end{equation*}

\begin{defn}\label{def:TVC} 
    A $\cont{2}$-function $\mtvc:\sspace\times \mathcal{X}\to \Ri{N}$,  $N\ge n-1$, is said to contain  a \emph{valid set of transverse coordinates} for the curve \eqref{eq:NomOrbit} if it vanishes on it, i.e. $\mtvc(s,\xs(s))\equiv 0$, and  for all $s\in\sspace$ it satisfies $\rank{\frac{\partial \mtvc}{\partial x}(s,\xs(s))}=\min(N,n)$ and $\rank{D\mtvc(s,\xs(s))}=n-1$. \qed
\end{defn}
For the case $N=n-1$, we will refer to $\mtvc$ as a  \emph{minimal} set of transverse coordinates by the fact that the mapping $(\mtvc,s)\mapsto x$ is then a diffeomorphism in some non-zero neighbourhood of $\eta_*$. One the other hand, whenever  $N\ge n$, we will refer to them as \emph{excessive} coordinates.

Note that the reason we consider the specific form $\mtvc=\mtvc(s,\state)$ rather than just $\mtvc=\mtvc(\state)$ is to highlight the possible dependence of these coordinates upon the choice of projection operator. That is to say,  given two different projection operators $\prj^1 (\cdot)$ and $\prj^2 (\cdot)$ for the curve, then, by a slight abuse of notation, the coordinates $\mtvc^1:=\mtvc(\prj^1(\state),\state)$ and $\mtvc^2:=\mtvc(\prj^2(\state),\state)$ are  not equivalent as long as $\frac{\partial }{\partial s}\mtvc\neq 0$, but  are nevertheless both  valid transverse coordinates for the curve.
  
With this in mind, suppose  $\mtvc$ is a  valid set of coordinates by Def.~\ref{def:TVC}. 
Differentiating, we find that their dynamics are described by
\begin{equation}\label{eq:minTVCdynamics}
    \Dmtvc=D\mtvc(s,\state)\left[f(\state)+g(\state)u\right].
\end{equation}
Our task will now be to linearize the dynamics of $\mtvc$ along the orbit $\eta_*$ in order to obtain a linear (periodic) system, the so-called linearized transverse dynamics, which we then can use to  design orbitally stabilizing feedback.  
Towards this end, we observe that since  $\mtvc(s,\xs(s))\equiv 0$,  we must have $\Dmtvc(s,\xs(s))\equiv 0$. Therefore,  by  defining 
\begin{equation*}
\Pis(s):=\frac{\partial \mtvc}{\partial \state}(s,\xs(s)), \end{equation*}
 it is implied that the following relation must hold:
\begin{equation}\label{eq:dyds}
    \frac{\partial \mtvc}{\partial s}(s,\xs(s))=-\Pis(s)\xs'(s).
\end{equation}
 Thus, sufficiently close to the orbit, it is true that
\begin{align*}
    \delta \mtvc&=D\mtvc(s,\xs(s))\delta\state= \Pis(s)\Oms(s)\delta \state,
\end{align*}
and hence, by \eqref{eq:deltaGTVCdeltaState}, we obtain
\begin{equation}\label{eq:dmtvcPiTvc}
    \delta \mtvc=\Pis(s)\delta \tvc.
\end{equation}
This naturally leads us to the following  unsurprising  statement, which simply shows  that there is a certain stability equality between  all sets of transverse coordinates.
\begin{prop}\label{prop:minTVCandExcTVC}
    The origin of a valid set of transverse coordinates $\mtvc$ is asymptotically stable if, and only if, the origin of the coordinates $\tvc$ is asymptotically stable.  \if\withproofs0{\qed}\fi
\end{prop}
\if\withproofs1
The proof of Proposition~\ref{prop:minTVCandExcTVC} can be found in \ref{app:proofminTVCandExcTVC}.
\fi

\subsection{Transverse linearization}
Now, let $\Psi(s):=D\mtvc(s,\xs(s))$ and consider the differentiable matrix function $\Pisin:\sspace\to\Ri{n\times N}$, defined by
\begin{equation}\label{eq:inverseofMinimalJacobian}
   \hspace{-1mm} \ \Pisin(s):=
    \begin{cases}
        \Oms(s)\Psi\transp(s){[\Psi(s)\Psi\transp(s)]}\inv & \hspace{-3mm}\text{if} \ N=n-1, \\
        \Pis\inv(s) & \hspace{-3mm}\text{if} \ N=n ,
        \\
       {[\Pis\transp(s)\Pis(s)]}\inv \Pis\transp(s) & \hspace{-3mm}\text{if} \ N>n .
    \end{cases}
\end{equation} 
This allows us to state the main result of this  section.
\begin{thm}\label{theorem:TVLGencoords} 
    Let $\mtvc\in\Ri{N}$ be a valid  set of transverse coordinates together with a   projection operator $\prj(\cdot)$. Then the linearization of their dynamics  \eqref{eq:minTVCdynamics} evaluated along the solution \eqref{eq:NomOrbit} is described by the constrained (differential-algebraic)
    linear-periodic system
    \begin{align}\label{eq:GentvcTVL} \nonumber
        \frac{d}{dt}\delta \mtvc &=\big[\Pis(s)A_\perp(s)+\Xi(s) \big]\Pisin(s)\delta \mtvc+\Pis(s)B_\perp(s)u \\ 
        0&=\DPs(s)\Pisin(s)\delta \mtvc
    \end{align} 
    where
    \begin{align*}
        A_\perp(s)&:=\Oms(s)A(s)-\xs'(s){\xs'}\transp(s) \DDPs(\xs(s))\nvel(s) 
        \\
        \Xi(s)&:=\nvel(s)\frac{\partial }{\partial \state }\left[\frac{\partial \mtvc}{\partial \state}(s,\state)\xs'(s)+\frac{\partial \mtvc}{\partial s}(s,\state)\right]\Bigg\lvert_{x=\xs(s)}
        \\
        B_\perp(s)&:=\Oms(s)B(s)
    \end{align*}
     given $A(s):=Df(\xs(s))$, $B(s):=g(\xs(s))$, $\nvel(s):=\DPs(s)f(\xs(s))$ and  with   $\Pisin(\cdot)$ as defined in \eqref{eq:inverseofMinimalJacobian}.\footnote{{Since $\xs:\sspace\to\eta_*$ is a regular parameterization, and thus $\nvel(s):=\DPs(s)f(\xs(s))>0$, it can be useful to note that one can utilize the fact that $ \frac{d}{ds}\delta \mtvc=\frac{1}{\nvel(s)}\frac{d}{dt}\delta\mtvc$ in order to solve \eqref{eq:GentvcTVL}.    }}
\end{thm}
\if\withproofs1
The proof of Theorem~\ref{theorem:TVLGencoords} is given in \ref{app:ProofOfGenTVL}. 
\fi

While there exists  several known explicit expressions for transverse linearizations  in the literature  (see e.g. \cite[Proposition 1.4]{hauser1994converse}, \cite[Theorem 12]{mohammadi2018dynamic}, \citep[Theorem 2]{shiriaev2010transverse}, \cite[Equation (4.23)]{leonov1995local}), they are all only valid for a specific class of coordinates or for specific choices of the projection operator.  Theorem~\ref{theorem:TVLGencoords}, on the other hand, provides explicit expressions valid for any set of transverse coordinates, and just as importantly, for  any choice of the projection operator. 
Also note that, while  Theorem 12 in \cite{mohammadi2018dynamic} provides equivalent expressions for the case when $N=n-1$, the proof of their statement is only valid whenever $\theta(s)$, as defined  in \eqref{eq:GenDpsDef},  is exactly zero for all $s\in \sspace$. This is due to their use 
of the pseudo-inverse of $\Psi$ as $\Pisin$, i.e. $\Pisin(s)=\Psi\transp(s){[\Psi(s)\Psi\transp(s)]}\inv$ (cf. $dH_{\varphi(\vartheta)}^{\dagger}$ therein). While that requires  $\Oms(s)=\Oms\transp(s)$ for $\DPs(s)\Pisin(s)\delta \mtvc=0$ to hold, and thus also the relation $\delta \state=\Pisin(s)\delta \mtvc+\xs'(s)\delta s$ between the differentials, it is here  satisfied directly by the slight modification of $\Pisin$ as  given by \eqref{eq:inverseofMinimalJacobian}.

To see the equivalence between the expression given in \cite[Theorem 12]{mohammadi2018dynamic} and \eqref{eq:GentvcTVL} for $N=n-1$, it is enough to note that for a $\mathcal{C}^1$-mapping $f_\perp:\Ri{n}\to\Ri{n-1}$ satisfying $f_\perp(\xs(s))\equiv 0$ for all $s\in\sspace$, then $Df_\perp(\xs(s))\xs'(s)\equiv 0$, and hence $Df_\perp(\xs(s))\Oms(s)=Df_\perp(\xs(s))$.  
\begin{cor}\label{cor:minTVL}
 Le $\mtvc:\Ri{n}\to\Ri{n-1}$ be a valid (minimal)  set of transverse coordinates defined independently of a projection operator, $\prj(\cdot)$, that is $\mtvc=\mtvc(\state)$. Then the linearization of their dynamics  \eqref{eq:minTVCdynamics} evaluated along the solution \eqref{eq:NomOrbit} is described by the    linear-periodic system
    \begin{align}\label{eq:mintvcTVL} \
        \frac{d}{d\mg}\delta \mtvc =\frac{1}{\nvel(s)}\left[Df_\perp(\xs(s))\Psi\pinv(s)\delta \mtvc+g_\perp(\xs(s))u \right]
    \end{align} 
    where $f_\perp(\state):=D\mtvc(\state)f(\state)$, $g_\perp(\state):=D\mtvc(\state)g(\state)$, $\nvel(s):=\DPs(s)f(\xs(s))$  and $\Psi\pinv(s)$ is the pseudo- (Moore--Penrose) inverse of $D\mtvc(\xs(s))$.\qed
\end{cor}

As stated in the introduction, the importance of  Theorem~\ref{theorem:TVLGencoords}, or equivalently Corollary~\ref{cor:minTVL},  is due to the exponential stability of linearized transverse dynamics implying asymptotic stability of the orbit. The convergence to $\eta_*$, however, does of course not mean that the system will be in phase with the nominal solution $\state_*(t)$. This well-known phase-shift property of orbital stability can be easily derived from the following statement, whose  proof is straightforward and follows the same lines as the proof of Theorem~\ref{theorem:TVLGencoords}.
\begin{lem}\label{lem:FAphaseshift}
Given a projection operator $\prj(\cdot) $ and transverse coordinates $\mtvc\in\Ri{N}$,  the linear-periodic system
\begin{equation*}
    \frac{d}{dt}\delta {\psi}=\DPs(s)\left[Df(\xs(s))-A_\perp(s)\right]\Pisin(s)\delta \mtvc+\DPs(s) g(\xs)(s)u
\end{equation*}
is the first approximation system of the dynamics of $\psi:=\int_0^t (\dot{s}(\sigma)-\nvel(s(\sigma)))d\sigma$ along $\eta_*$.
\end{lem}

It is therefore clear that  $\dot{s}\equiv\nvel(s)$ if $\mtvc\equiv u\equiv 0$, showing that the system might not be  in phase with the nominal solution after convergence to the orbit.  Moreover, it is evident that this system does not influence the stability of the orbit, such that one only needs to consider the linearized transverse dynamics in this regard.

While analysis of the system \eqref{eq:mintvcTVL} given a minimal set of coordinates is quite straightforward, one must take into account the transversality constraint when considering an excessive set of coordinates in \eqref{eq:GentvcTVL}. Thus, in order to provide  further insight into the  transverse linearization of an excessive set of coordinates, we will focus on a specific set of orthogonal coordinates in the sequel.

%% file: Generic_set_of_Excessive_Ort_Coors.tex
\section{A Generic set of Excessive Orthogonal  Coordinates}\label{sec:OrtCoords}
Consider again the excessive coordinates   previously defined in \eqref{eq:TVC}, namely $\tvc:=\state-\xs(s)$. Using the fact that $\DP(\cdot)$ and $f(\cdot)$ are assumed to be $\mathcal{C}^1$, one may use their first-order Taylor expansions about $\eta_*$ in order to show that the transverse dynamics \eqref{eq:TVD1} then can be rewritten as
\begin{equation}\label{eq:LTVD}
    \Dtvc=A_\perp(s)\tvc+\Omega(\state)g(\state) \ac+\Delta(s,\tvc),
\end{equation}
where $\|\Delta(\cdot,\tvc)\|=O(\|\tvc\|^2)$, that is
\begin{equation*}
    \lim_{\tvc\to 0}\frac{\|\Delta(\cdot,\tvc)\|}{\|\tvc\|}=0.
\end{equation*}
The choice of notation in  Theorem~\ref{theorem:TVLGencoords} thus becomes clear by its following corollary. 
\begin{cor}\label{cor:TVLforNatExCoords}
 The constrained linear-periodic system
\begin{equation}\label{eq:TVLsys}
    \frac{d}{dt}\delta \tvc=A_\perp(s)\delta \tvc+B_\perp(s)u, \quad \DPs(s)\delta \tvc=0, 
\end{equation}
corresponds to the linearization along \eqref{eq:NomOrbit}  of the dynamics  of the excessive set of coordinates defined in \eqref{eq:TVC}.
\qed
\end{cor}

As previously stated, the coordinates $\tvc$ will depend upon the choice of $\prj(\cdot)$. While there  in general will exist  many valid candidates for this projection operator, all with different properties and resulting in  different transverse hypersurfaces (moving Poincaré sections) on which the coordinates $\tvc$ evolve, we will from now on consider those satisfying the orthogonality condition:
\begin{equation}\label{eq:OrthogonalityCondition}
    {\xs'}\transp (s)\tvc\equiv 0.
\end{equation}
Note that this is locally equivalent to  $s=\argmin_{s\in\mathcal{S}}\|\state-\xs(s)\|^2$, and so the  Jacobian of this $\prj(\cdot)$ is  given by
\begin{equation}
    \DP(\state)=\frac{{\xs'}\transp(s)}{\|\xs'(s)\|^2-{\xs''}\transp(s)\left(\state-\xs(s)\right)},
\end{equation}
while, moreover, it can be shown     that  $\Delta(\cdot)$ then satisfies ${\xs'}\transp(s) \Delta(s,\tvc)\equiv 0$ \citep{leonov2006generalization}.
In addition, using   \eqref{eq:OrthogonalityCondition} and that 
$\DDP(\xs(s))\nvel(s)\tvc=\frac{\xs'(s){\xs'}\transp(s)}{\|\xs'(s)\|^4}A\transp(s) \tvc$, the matrix function $A_\perp(\cdot)$ can then be simplified to
\begin{equation}\label{eq:AperpSimp}
    A_\perp(s):=\Oms(s)A(s)-\frac{\xs'(s){\xs'}\transp(s)}{\|\xs'(s)\|^2}A\transp(s).
\end{equation}
 Thus the {linearized transverse dynamics} are  given according to Corollary~\ref{cor:TVLforNatExCoords} with  \eqref{eq:AperpSimp} and  $\DPs(s)=\frac{{\xs'}\transp(s)}{\|\xs'(s)\|^2}$,

Note that the coordinates \eqref{eq:TVC} together with the   orthogonality condition \eqref{eq:OrthogonalityCondition} have been considered several times times before in relation to the study of the (in-)stability of solutions of autonomous dynamical systems (see e.g. \cite{borg1960condition,hartman1962global,zubov1999theory,leonov2006generalization,hauser1994converse}). However, they have not, to our best knowledge, been used together  for the purpose  of designing orbitally stabilizing feedback controllers for nonlinear systems of the form \eqref{eq:DynSYs}.
For this purpose, however, 
  the relation ${\xs'}\transp(s) \delta \tvc\equiv 0$  is of particular interest. This is because, unlike a minimal set of coordinates in which the transversality condition $\DPs(s)\Pisin(s)\delta \mtvc=0$ in \eqref{eq:GentvcTVL} is satisfied directly through $\Pisin$, it must be satisfied through the coordinates themselves for an excessive set. 

%% file: LimitationsOfTVL.tex
\subsection{Limitations of the excessive transverse linearization}\label{sec:LimTVLD}
Consider the linear system 
\begin{equation}\label{eq:TVLcompsys}
\dot{y}=A_\perp(s)y+B_\perp(s)u
\end{equation}
corresponding to  \eqref{eq:TVLsys},  with $A_\perp$ as in \eqref{eq:AperpSimp}  but without the transversality condition ${\xs'}\transp(s) y\equiv 0$. 
It can be shown that the undriven system  $(u\equiv 0$)  then has the solution
\begin{equation}\label{eq:NonVanSol}
    y_{\parallel}=\frac{\xs'(s)}{\|\xs'(s)\|^2\nvel(s)}=\frac{\xs'(s)}{{\xs'}\transp(s) f_s(s)},
\end{equation}
whose characteristic exponent\footnote{
The number (or the symbols, $\pm\infty$), given by  the formula $\limsup_{t\to+\infty}\frac{1}{t}\ln\|x(t)\|$ is called the characteristic exponent of the continuous function $x:[0,\infty)\to\Ri{n}$ \citep{leonov2006generalization}.
} 
evidently is exactly zero. Moreover, an additional $(n-1)$  linearly independent solutions of the undriven system can be found, which we denote $y^1_\perp(\cdot),\dots,y^{n-1}_\perp(\cdot)$, and which form a basis of the kernel  of $\DPs(s)$ for a given $s\in\sspace$ (it can be shown that $\frac{d}{dt}(y_{\parallel}\transp y_\perp^i)\equiv 0$), and hence satisfy  condition \eqref{eq:OrthogonalityCondition}.  Using these solutions,  let $\Pperp(s)=[\pperp^1(s),\dots,\pperp^{n-1}(s)]$ denote a smooth normalized basis of the kernel of $\DPs(s)$, with $\pperp^i(\cdot)$ defined by $\pperp^i(s(t))=y^i_\perp(t)/\|y^i_\perp(t)\|$, and let $\Pperpinv$ denote its pseudo-inverse, that is $\Pperpinv:={(\Pperp\transp\Pperp)}\inv \Pperp\transp$.

Consider now the first approximation (variational) system of \eqref{eq:DynSYs} along the curve \eqref{eq:NomOrbit}:
\begin{equation}\label{eq:FirstApprox}
    \frac{d}{dt}\delta \state=A(s)\delta \state+B(s)u.
\end{equation}
  The following statement can then be seen as analogous to the Andronov--Vitt theorem for the system \eqref{eq:TVLcompsys}.
\begin{prop}\label{prop:limETL}
    The system \eqref{eq:TVLcompsys} has $(n-1)$ linearly independent solutions of the  form $\Pperp(s(t))\xi_\perp(t)$ with $\xi_\perp\in\Ri{n-1}$ a solution to the $(n-1)$-dimensional system 
    \begin{equation}\label{eq:TVLsubSys}
    \dot{\xi}_\perp= \Pperp\transp(s) A(s) \Pperp(s)\xi_\perp+\Pperp^\dagger(s) B(s)(s)u.
    \end{equation}
    In addition, it 
    has a solution with a non-vanishing part in the direction of  \eqref{eq:NonVanSol}   regardless of the control input $u$. \qed
\end{prop}
\if\withproofs1
The proof of Proposition~\ref{prop:limETL} is stated in \ref{app:ProofOfLim}. 
\fi

An important consequence of Proposition~\ref{prop:limETL}  is the fact that the origin of the system \eqref{eq:TVLcompsys} can never be asymptotically stabilized. That is to say, even if one can find some  feedback asymptotically stabilizing the origin of the system \eqref{eq:TVLsys}, and consequently the periodic orbit, the system \eqref{eq:TVLcompsys} will regardless have a non-vanishing solution whose characteristic exponent is  zero.  Thus the  usefulness of this system in terms of control design is  limited due to its non-stabilizable subspace.   On the other hand, we can infer that if the pair
$(\Pperp\transp A\Pperp,\Pperp^\dagger B\Pperp)  $
 is stabilizable, then we can  stabilize the orbit utilizing some controller designed to stabilize the subsystem \eqref{eq:TVLsubSys}. 
The obvious  alternative is therefore to  try to directly stabilize this subsystem. Yet, this requires knowledge of the  basis $\Pperp(\cdot)$.

Clearly it would instead be beneficial to find some way of stabilizing the subsystem \eqref{eq:TVLsubSys} without the need to form  $\Pperp(\cdot)$. 
In this regard, we will  introduce next a linear  comparison system of  \eqref{eq:TVLcompsys}, for which, under conditions we state in Proposition~\ref{theorem:mainResult}, the asymptotic stability of its origin implies asymptotic stability of the origin of the subsystem \eqref{eq:TVLsubSys} and consequently the asymptotic orbital stability of the nominal solution.

%% file: ComparisonSys.tex
\subsection{The existence of a comparison system}\label{sec:AuxSys}

Suppose we left-multiply both sides of \eqref{eq:TVLcompsys} by the matrix function $\Oms(s)$. Utilizing its properties  (see Lemma~\ref{lemma:Omegas}), one can then rewrite the system on several different equivalent forms,  with the following among them:
\begin{equation}\label{eq:CompSysInit}
    \Oms(s)\left[\dot y -\Oms(s)\left(A(s)y+B(s)u\right)\right]=0.
\end{equation}
Consider, therefore, the   linear-periodic  system
\begin{equation}\label{eq:CompSys}
    \dot{w}=\Oms(s)A(s)w+\Oms(s)B(s)v, \ w\in\Ri{n}, \ v\in\Ri{m},
\end{equation}
corresponding to the terms inside the brackets of the descriptor system \eqref{eq:CompSysInit} being set to zero.
Roughly speaking, we will show that if there exists a  feedback of the form $v=K(s)w$ which ``sufficiently" stabilizes the origin of this \emph{comparison system}, then the controller $u=K(s)\delta \tvc$ stabilizes the origin of the linearized transverse dynamics \eqref{eq:TVLsys} as well. Thus this comparison system  can allow one to  find a stabilizing feedback for \eqref{eq:TVLsys}  without the need to circumvent the uncontrollable subspace always present in \eqref{eq:TVLcompsys} and without having to compute the Hessian $\DDP(\cdot)$.
Indeed, there are several connections between these systems, such as the following  spectrum condition.
\begin{lem}\label{lemma:spectrumCond}
    Consider the system \eqref{eq:DynSYs} with  the feedback $u=K(\prj(\state))[\state-\xs(\prj(x))]$ for some  Lipschitz continuous matrix function $K:\sspace\to \Ri{m\times n}$. Then the (minimal) sum of  the characteristic exponents of  the  systems \eqref{eq:TVLcompsys}, \eqref{eq:FirstApprox} and \eqref{eq:CompSys} are the same. \qed
\end{lem}
\if\withproofs1
The proof can be found in  Appendix~\ref{app:SpecCondProof}.
\fi

    
Suppose, therefore, that a (Lipschitz continuous) matrix function $K:\sspace\to \Ri{m\times n}$ exists such that the largest characteristic exponents, $\lambda_M$, of the closed-loop system
\begin{equation}
    \dot{w}=\Omega(s)\left(A(s)+B(s)K(s)\right)w
\end{equation} 
satisfies $\lambda_M<0$; i.e. we assume \eqref{eq:CompSys} is stabilizable.  Let $W(t)$ denote the state transition (Cauchy) matrix for this system. Then, by a small modifications of theorems 2 and 4 in \cite{leonov2007time}, there exists
some number $C>0$ and a scalar functions $\zeta:[0,\infty) \to \Ri{}$ satisfying
\begin{equation}
    \lim_{t\to \infty} \frac{1}{t}\int_\tau^t \zeta(\sigma) d\sigma =\lambda_M \quad \forall \tau \ge 0,
\end{equation}
such that the following inequality
\begin{equation}\label{eq:STMbounds}
    \|W(t)W\inv(\tau)\|\le C\exp\left(\int_\tau^t \zeta(\sigma) d\sigma\right) \ \ \forall t\ge \tau \ge 0
\end{equation} 
is satisfied. The main result of this section follows.
\begin{prop}\label{theorem:mainResult} 
Let $\prj(\cdot)$ be taken as to satisfy \eqref{eq:OrthogonalityCondition}.
Suppose that $\|A(s)\|\le \alpha$ for all $s\in\sspace$
and that the  inequality
\begin{equation}\label{eq:MRcond}
    \lambda_M<-C\alpha\le 0
\end{equation}
holds. Then the controller $u=K(s)\tvc$  with $s=\prj(x)$ asymptotically stabilizes the origin of the system \eqref{eq:LTVD} and consequently renders the periodic solution of the dynamical system \eqref{eq:DynSYs}  asymptotically orbitally stable. 
\end{prop}
\if\withproofs1
The proof of this statement is given in \ref{app:ProofMR}. 
\fi 

\begin{rem}
    The value of the above statement is not in the condition  \eqref{eq:MRcond} per se. Rather, its importance is simply due to the fact that it shows the possibility of orbitally stabilizing the solution by designing a stabilizing feedback for the comparison system \eqref{eq:CompSys}. Indeed, the condition \eqref{eq:MRcond} is by no means unique,   and similar conditions can be stated using, for example, Lyapunov's second method.\qed 
\end{rem}

    It is  also of practical importance to note that if a controller $v=K(s)w$ stabilizing the origin of the comparison system \eqref{eq:CompSys} has been designed, then one does not need to check the conditions of the theorem. That is to say, one  can instead utilize the Andronov--Vitt theorem on the first approximation system $\delta\Dstate=(A(s)+B(s)K(s)\Oms(s))\delta \state$ to validate that it will also be a stabilizing controller for \eqref{eq:TVLsys}; or, equivalently, check that the  system \eqref{eq:TVLcompsys} has $(n-1)$ characteristic multipliers within the unit circle.
    As yet another alternative, one can utilize the following.
    \begin{lem}\label{lemma:AVcompSys}
        If the  system \eqref{eq:CompSys} under the controller $v=K(s)\Oms(s)w$ 
        has  one  simple  zero  characteristic  exponent and the remaining $(n-1)$ characteristic exponents have strictly negative real parts,
        then the controller $u=K(s)\tvc$ asymptotically stabilizes the origin of the system \eqref{eq:LTVD}. \qed
    \end{lem}
\if\withproofs1    
 Indeed, it is not difficult to see that the above is  implied by \eqref{eq:chiSys} in the proof of Proposition~\ref{theorem:mainResult} to be a sufficient condition for the asymptotic stability of the subsystem \eqref{eq:TVLsubSys}. 
 \fi
 This again shows that one does not need to compute the Hessian of $\prj(\cdot)$ in order to validate the stability of the orbit.  Moreover, this has an additional advantage  compared to the Andronov--Vitt theorem  arising whenever the dynamical system has a periodic  solution only in the presence of  some  non-zero nominal control input $\upsilon(s(t))\equiv u_*(t)$, i.e.  $\frac{d}{dt}\xs(s)=f(\xs(s))+g(\xs(s))\upsilon(s)$.
    As then the matrix $A(\cdot)$ of the first approximation  is  given by
    \begin{equation*}
        A(s)=\left.\left[\frac{\partial f}{\partial x}+g\upsilon'(s)\DPs(s)+\sum_{i=1}^m\frac{\partial g_i}{\partial x}\upsilon_i(s)\right]\right\lvert_{x=x_s(s)},
    \end{equation*}
    one needs to compute $\upsilon'(s)$ in order to utilize the Andronov-Vitt Theorem, whereas it can be  omitted in the transverse linearization, and consequently for the comparison system \eqref{eq:CompSys}, due to the condition $\DPs(s)\delta z \equiv 0$. 
    
    We illustrate the above scheme in a simple example next. 

%% file: exampleBanHauser.tex
\section{Illustrative Example}\label{sec:ExampleBanHau}
Consider the system
\begin{subequations}
\begin{align}
    \Dot{x}_1 &=x_2+x_1x_3+x_1u \\
    \Dot{x}_2 &=-x_1+x_2x_3+x_2u \\
    \Dot{x}_3 &= u
\end{align}
\end{subequations}
which for $u\equiv 0$ has a  family of periodic orbits given by
\begin{equation}\label{eq:HauserOrbit}
    \eta_a=\{\state\in\Ri{3}| x_1^2+x_2^2=a^2,x_3=0,a>0\}.
\end{equation}
This system has previously been considered in   \cite{banaszuk1995feedback}, where a (transverse) feedback linearizing approach was utilized in order to find a minimal set of transverse coordinates. More specifically, they showed  that by taking $\theta=-\arctan(x_2/x_1)$, there exists a pair of  transverse coordinates $(\sigma_1,\sigma_2)$, defined as $\sigma_1:=\log \left({\sqrt{x_1^2+x_2^2}}\right)-\log(a) -x_3$ and $\sigma_2:=x_3$,
such that $(x_1,x_2,x_3)\mapsto (\theta,\sigma_1,\sigma_2)$ is a diffeomorphism everywhere except $(x_1,x_1)=(0,0)$. Moreover, the dynamics of $\theta$ is trivial ($\Dot{\theta}=1$) while  the dynamics of the transverse coordinates $(\sigma_1,\sigma_2)$ are linear: $\Dot{\sigma}_1=\sigma_2$, $  \Dot{\sigma}_2=u$.
While this is clearly  a convenient choice of coordinates,  and illustrates the possibility of finding a minimal set of coordinate that can greatly simplify control design, it also shows the challenge of finding a (convenient) set of  coordinates even for such a simple, low dimensional system. 

Let us therefore instead consider $s=\prj(\state)=\text{atan2}\left({x_1},{x_2}\right)$ with $\Dot{s}_*(t)=\nvel(s(t))=1$, which here satisfies the orthogonality condition \eqref{eq:OrthogonalityCondition} ($\text{atan2}(\cdot)$ denotes the four-quadrant arctangent function), and which lets us parameterize the orbit $\eta_a$ by 
$\state_s(s)=[a\sin(s), a\cos(s) , 0]\transp$. The linearized transverse dynamics \eqref{eq:TVLcompsys} then  becomes
\begin{equation}\label{eq:HauserLTD}
    \frac{dy}{ds} =\begin{bmatrix}0 & 1 & a\sin(s) \\ -1 & 0 & a\cos(s) \\ 0 & 0 & 0 \end{bmatrix}y + \begin{bmatrix} a\sin(s)\\ a\cos(s) \\ 1\end{bmatrix}u,
\end{equation}
while its comparison system \eqref{eq:CompSys} is given by
\begin{equation}\label{eq:HauserALTD}
    \frac{dw}{ds} =\begin{bmatrix}-\frac{\sin(2s)}{2} & \sin^2(s) & a\sin(s) \\ -\cos^2(s)
    & \frac{\sin(2s)}{2} & a\cos(s) \\ 0 & 0 & 0 \end{bmatrix}w + \begin{bmatrix} a\sin(s)\\ a\cos(s) \\ 1\end{bmatrix}v.
\end{equation}
 Taking  $a=1$, we designed a stabilizing controller for the comparison system \eqref{eq:CompSys}, in which  the found controller gains can be seen in Figure~\ref{fig:HauserCont}.  These gains correspond to the feedback matrix $K(s)=[k_1(s),k_2(s),k_3]=-B_\perp\transp(s) R(s)$ with $R(s)=R\transp(s)$  the positive definite  solution to the periodic Riccati differential equation 
\begin{align*}
    \frac{dR}{ds}{}+\Oms A\transp R+R\Oms A+I_{3}-RB_\perp B_\perp\transp R= 0. 
\end{align*}
 With this controller, the characteristic exponents of \eqref{eq:HauserLTD} were approximately $(0,   -1.73,   -1)$, implying the asymptotic stability of the orbit by Proposition~\ref{prop:limETL}; while for the system \eqref{eq:HauserALTD} they were approximately
 $(-0.86 \pm 0.5i,-1)$, showing it is indeed an orbitally stabilizing controller as we would expect from Proposition~\ref{theorem:mainResult}.
\begin{figure}
    \centering
    \includegraphics[width=0.8\linewidth]{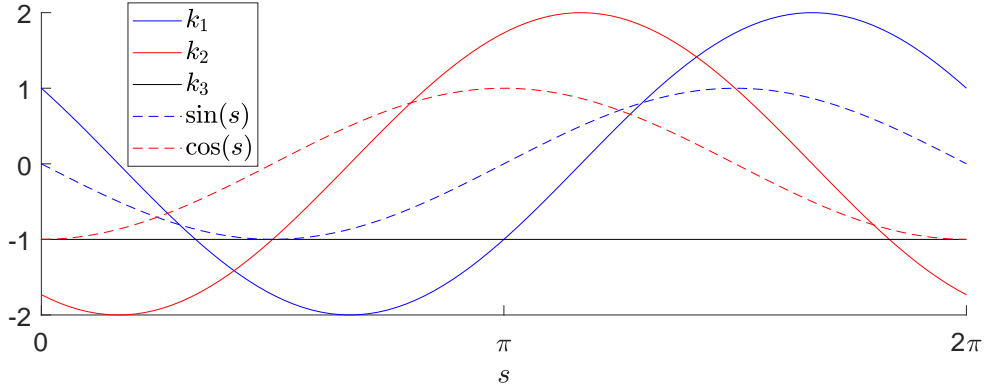}
    \caption{Controller gains stabilizing \eqref{eq:HauserALTD}. }
    \label{fig:HauserCont}
\end{figure}

Let us now also demonstrate a certain limitation of Proposition~\ref{theorem:mainResult} by instead considering
 the feedback
\begin{equation}\label{eq:SimpleCont}
    u(\tvc)=-\begin{bmatrix}\sin(s) & \cos(s) & 1\end{bmatrix}\tvc
\end{equation}
which  stabilizes the system \eqref{eq:HauserLTD}, and consequently asymptotically stabilizes the orbit \eqref{eq:HauserOrbit} for any $a>0$. More specifically,  it can be shown that the modified periodic Riccati differential equation 
\begin{align*}& \label{eq:RDE}
    \Oms\transp\Big[\frac{d{R}_\perp}{ds}+{A}_{\perp}\transp{R}_\perp+{R}_\perp{A}_{\perp}+\I{3}
    -{R}_\perp{B}_{\perp}{B}_{\perp}\transp{R}_\perp\Big]\Oms=0, 
\end{align*}
has a family of solutions given by
\begin{align*}
    R_\perp(s)=\Oms^i(s)\begin{bmatrix}\frac{1}{a}& 0& 0\\
    0 & \frac{1}{a} & 0 \\
    0 & 0 & 1
    \end{bmatrix}\Oms^j(s)
    +k\begin{bmatrix}\cos^2(s)& -\frac{\sin(2s)}{2}& 0\\
    -\frac{\sin(2s)}{2} & \sin^2(s) & 0 \\
    0 & 0 &  0
    \end{bmatrix}
\end{align*}
 for any $k\in\Ri{}$ and $i,j\in\{0,1\}$, such that \eqref{eq:SimpleCont} corresponds to $u(\tvc)=-B_\perp\transp(s) R_\perp(s)\tvc$. Therefore, by taking $V=\delta \tvc\transp R_\perp(s)\delta\tvc$, we have $\dot{V}\le-\|\delta\tvc\|^2$ implying the asymptotic stability of the nominal solution. On other hand, in accordance with Proposition~\ref{prop:limETL}, 
it can be shown that the closed-loop system, i.e. 
$A_{cl}(s):=A_\perp(s)-B_\perp(s)B_\perp\transp(s) R_\perp(s)$, without the orthogonality condition \eqref{eq:OrthogonalityCondition}
has the solution $\xs'(s)$ with characteristic exponent equal to zero. Its two other  independent solutions are $[0,0,e^{-t}]\transp$ and 
$\big[\sin(s(t)),\cos(s(t)),le^{(a-1)t}+\frac{e^{-at}}{a-1}
    \big]\transp$ with $ l\in\Ri{}$. Taking $l=0$, their characteristic exponents equals $-1$ and  $-a$, respectively, again implying the asymptotic stability of the nominal solution.

Consider now the comparison system \eqref{eq:HauserALTD} with the above controller, i.e.
\begin{equation*}
    v(w)=-\begin{bmatrix}\sin(s) & \cos(s) & 1\end{bmatrix}w.    
\end{equation*}
It  too has $\xs'(s)$ as a solution, while it can be shown that  $-1$ and $-a$ are the characteristic exponents of the two remaining independent solutions (although note these solutions are different to those of \eqref{eq:HauserLTD} given above). We can therefore utilize Lemma~\ref{lemma:AVcompSys} to validate that the controller is asymptotically orbitally stabilizing, but we cannot utilize Proposition~\ref{theorem:mainResult} for this purpose.  

So why is not the origin of the comparison system \eqref{eq:HauserALTD} asymptotically stable under the controller \eqref{eq:SimpleCont}?
It turns out that the existence of the solution $\xs'(s)$ is clear simply by  noticing that $B_\perp(s)B_\perp\transp(s) R_\perp(s)\equiv\hat{K}(s)\Oms(s)$ given
\begin{equation*}
    \hat{K}(s):=\begin{bmatrix}
    a & 0 & a\sin(s) \\ 0 & a & a\cos(s) \\
    \sin(s) & \cos(s) & 1 
    \end{bmatrix}.
\end{equation*}
Thus   $u(w)=K(s)w\equiv 0$ for any  $w\in \text{span}(\xs'(s))$. It  follows that a controller  asymptotically stabilizing the linearized transverse dynamics \eqref{eq:TVLsys} will not necessarily asymptotically stabilize the comparison system \eqref{eq:CompSys}.
On the other hand, it is quite interesting to note that   all the characteristic exponents of both the systems $\dot{\hat{y}}=\big(A_\perp(s)-\hat{K}(s)\big)\hat{y}$ and $\dot{\hat{w}}=\big(\Oms(s)A(s)-\hat{K}(s)\big)\hat{w}$ have strictly negative real parts and sum to $(-2a-1)$. 

%% file: conclusion.tex
\section{Concluding remarks }\label{sec:conclusion}
In this paper, we have provided analytical expressions of the linearized transverse dynamics of any valid (minimal or excessive) set of transverse coordinates. In addition, we have defined a generic set of  easy-to-compute  orthogonal coordinates and shown a certain equivalence between their stability and that of any other valid  set. It was further demonstrated that their origin  could be stabilized  by stabilizing a comparison system    of the linearized transverse dynamics. This of course relies on the stabilizability of this comparison system, such that conditions for its stabilizability, as well as  the connection to the stabilizability of the linearized transverse dynamics are topics  of interest and requiring further study. The presented approach nevertheless lays the foundations for further development and generalizations, such as, for example, its extension to hybrid dynamical systems and to  non-periodic motions.

%% file: appendixTVL.tex
\section{}
\subsection{Proof of Lemma~\ref{lemma:Omegas}}\label{app:proofPmegas}
   As $s=\prj(\xs(s))$ by definition~\ref{def:ProjOp}, \eqref{eqDpDxs} simply follows from
\begin{equation}\label{eq:nvelDef}
    \frac{d}{ds}\prj(\xs(s))=\frac{d}{ds}s=\DPs(s)\xs'(s)=1
    \end{equation} 
      Using \eqref{eqDpDxs}, it is then straightforward to validate that   $\DPs(s)$ and $\xs'(s)$  are left- and right annihilators of $\Oms(s)$, respectively, and that $\Oms(s)=\Oms^2(s)$. Lastly, as $\Oms$ consists of a rank $n$ matrix (i.e. $\I{n}$) and a rank one matrix (i.e. $\xs'\DPs$), as well as the existence of the annihilators, implying its kernel is of dimension one, it follows that its rank is always $(n-1)$ by the rank-nullity theorem.     
\subsection{Proof of Proposition~\ref{prop:minTVCandExcTVC}}\label{app:proofminTVCandExcTVC}
To prove Proposition~\ref{prop:minTVCandExcTVC} we need only show that $\|\delta \mtvc \|= 0$ if and only if $\|\delta \tvc\| = 0$.
In fact, because of \eqref{eq:dmtvcPiTvc} it is essentially a corollary of the following statement.
\begin{lem}\label{lem:PisdeltaTvc}
 Let $\varphi_\perp(s)\in\ker{\DPs(s)}$, then $\|\Pis(s)\varphi_\perp(s)\|\equiv 0$ if and only if $\|\varphi_\perp(s)\|\equiv 0$.
\end{lem}
\begin{proof}
 It is trivially true for $N\ge n$ as $\rank{\Pis(s)}=n$ by Def.~\ref{def:TVC}.  For $N=n-1$, we have $\rank{\Pis(s)}=\rank{\Pis(s)\Oms(s)}=n-1$. Because $\ker{\Pis(s)\Oms(s)}=\vspan{\xs'(s)}$, it is  implied that $\ker{\Pis(s)}\cap\im{\Oms(s)}=\{0\} $,  i.e., the kernel of $\Pis(s)$ does not lie the image (range) of $\Oms(s)$. The statement  then follows from the fact that  $\Oms(s)\varphi_\perp(s)=\varphi_\perp(s)$ for any $\varphi_\perp(s)\in\ker{\DPs(s)}$. 
\end{proof}
It is here important to note that this does not imply that $\Pis(s)\xs'(s)=0$, or equivalently $\Pis(s)\Oms(s)=\Pis(s)$, is true in general. For example, consider $\state\in\Ri{2}$ with $\xs(s)=[s,\nvel(s)]\transp$ for some smooth function $\nvel(s)>0$. We can then simply take $\prj(\state)=[1,0]x$ and $\mtvc=[0,1](x-\xs(s))$.  It therefore follows that $\Pis(s)=[0,1]$, and so $\DPs\transp(s)=[1,0]\transp\in\ker{\Pis(s)}$. Although  note that $\Pis(s)\xs'(s)=0$ will of course always be true  whenever $\mtvc$ is defined independently of the parameterization, i.e. 
if $\partial \mtvc/\partial s=0$,  or if $\prj(\cdot)$ is defined such that $\theta(s)$ in \eqref{eq:GenDpsDef} is always zero. 
\subsection{Proof of Theorem~\ref{theorem:TVLGencoords}}\label{app:ProofOfGenTVL}
We begin by making the following claim.
\begin{lem}\label{lemma:Pisin}
Let $\Pisin$ be defined according to \eqref{eq:inverseofMinimalJacobian}. Then $\Pisin$ is of rank $\min(n,N)$ and the relation
\begin{equation}\label{eq:deltaTvcdeltaMtvc}
    \delta \tvc=\Pisin(s)\delta \mtvc
\end{equation}    
is valid   for all $s\in\sspace$ as long as $\delta \mtvc$ satisfies  the    condition 
\begin{equation*}
    \DPs(s)\Pisin(s)\delta \mtvc=0    ,
\end{equation*}
which is necessary for $N\ge n$.
\end{lem}
\begin{proof}
For $N\ge n$, the definition of $\Pisin(s)$ follows directly from \eqref{eq:dmtvcPiTvc}. However, as then $\rank{\Pisin(s)}=n$, and by the fact that $\DPs(s)\delta \tvc=0$, we obtain the requirement that  $\delta \mtvc$ then must satisfy the condition $\DPs(s)\Pisin(s)\delta \mtvc=0$.
For $N=n-1$, the necessary condition $\DPs(s)\Pisin(s)\delta \mtvc=0$ cannot be satisfied by the coordinates as they then are independent quantities. As $\DPs(s)$ is a left-annihilator of $\Oms(s)$, this will always be satisfied for  $\Pisin(s)$ taken according to \eqref{eq:inverseofMinimalJacobian}. Moreover, as then $\rank{\Phi}(s)=n-1$ for all $s\in\sspace$, it follows that $\Phi(s)\Phi\transp(s)$ is invertible.   
Therefore, using  $\Oms^2(s)=\Oms(s)$, we only need to show that $\rank{\Oms(s)\Oms\transp(s)\Pis\transp(s)}=n-1$, which by Lemma~\ref{lemma:Omegas} is equivalent to showing that there does not exists a  vector $w\in\Ri{n}$ such that $\Oms\transp(s)w=\xs'(s)$.  But as  $\im{\Oms\transp(s)}=\ker{\xs'}\transp(s)$, no such vector can exist. \qed
\end{proof}
We now note that \eqref{eq:minTVCdynamics} is affine in the control input $u$, from which the term $\Pis(s)B_\perp(s)$ and the definition of the matrix function $B_\perp(s)$ naturally follows. In order to find the remaining terms, we define,  from the remaining  part of the right-hand side of \eqref{eq:minTVCdynamics}, the function $F(s,x):={D\mtvc}(s,x)f(x)$. As we can write  its differential about $\eta_*$ both in terms of  variations in the states $\delta\state$ and in the coordinates $(\delta s,\delta \mtvc)$, we obtain the following relations:
\begin{equation*}
    \delta F_s=\frac{\partial F_s}{\partial \mtvc}\delta \mtvc+\frac{\partial F_s}{\partial s}\delta s=\frac{\partial F_s}{\partial \state}\delta \state,
\end{equation*}
where we have used the subscript notation $F_s:=F(s,\xs(s))$.
But as $F_s=0$, we must also  have $\frac{d}{dt}F_s=0$, and hence
\begin{equation*}
    \frac{d}{dt}F_s=\frac{\partial F_s}{\partial \mtvc}\Dmtvc(s,\xs(s))+\frac{\partial F_s}{\partial s}\nvel(s)=\frac{\partial F_s}{\partial s}\nvel(s)=0.
\end{equation*}
This implies $\partial F_s/\partial s =0$, and
therefore
\begin{equation}\label{eq:MatrixEq}
\frac{\partial F_s}{\partial \mtvc}\delta \mtvc=\frac{\partial F_s}{\partial \state}\delta \state.    
\end{equation}

We now note that by \eqref{eq:dyds} we have
\begin{equation}\label{eq:deltaMtvc}
    \delta \mtvc=\Pis(s)\delta \state-\Pis(s)\xs'(s)\delta s.
\end{equation}
On the other hand, by \eqref{eq:inverseofMinimalJacobian}, it can be shown that
\begin{equation}\label{eq:deltaStatedeltaMtvc}
    \delta \state=\Pisin(s)\delta \mtvc+\xs'(s)\delta s,
\end{equation}
which  follows directly from \eqref{eq:deltaMtvc} and the definition of $\Pisin(s)$ whenever $N\ge n$, while for $N=n-1$ it follows from 
\begin{equation*}
    \begin{bmatrix} \DPs(s) \\ \Pis(s)\Oms(s) \end{bmatrix}\begin{bmatrix} \xs'(s) & \Pisin(s) \end{bmatrix}=\I{n}.
\end{equation*}
Thus inserting for $\delta \state$  from \eqref{eq:deltaStatedeltaMtvc} in \eqref{eq:MatrixEq} and noting that $\frac{d}{dt}F_s=\frac{\partial F_s}{\partial \state}\xs'(s)\nvel(s)=0$, we obtain
\begin{equation} \label{eq:MatEq}
    \frac{\partial F_s}{\partial \mtvc}\delta \mtvc=\frac{\partial F_s}{\partial \state}\Pisin(s)\delta \mtvc.
\end{equation}
Straightforward computations then  show that $\frac{\partial F_s}{\partial \state}=\Pis(s)A_\perp(s)+\Xi_s(s)+\frac{\partial F_s}{\partial s}\DPs(s)$, in which, when inserted into \eqref{eq:MatEq}, the  term $\frac{\partial F_s}{\partial s}\DPs$ can be omitted  due to the    the condition $\DPs(s)\Pisin(s)\delta \mtvc=0$. 
Indeed, the condition $\DPs(s)\Pisin(s)\delta \mtvc=0$  is in itself necessary for $N\ge n$ by Lemma~\ref{lemma:Pisin}  and must therefore be added in order to restrict the solutions of the linear  system to  the transverse plane. Thus the system \eqref{eq:GentvcTVL}  follows.

 \subsection{Proof of Proposition~\ref{prop:limETL}}\label{app:ProofOfLim}
 Consider the matrix function 
 \begin{equation*}
    U(s):=\begin{bmatrix}\frac{\xs'(s)}{\|\xs'(s)\|} & \Pperp(s)\end{bmatrix},     \quad 
    U\inv=\begin{bmatrix}\frac{\xs'}{\|\xs'\|} & \left(\Pperp\pinv\right)\transp
    \end{bmatrix}\transp, 
\end{equation*}
which  allows us to rewrite  $\Oms(s)$  on the form
\begin{equation*}
\Oms(s)=U(s)\Lambda U\inv(s), \quad \Lambda:=\begin{bmatrix}0 & \0{1\times n-1} \\ \0{ n-1 \times 1} & \I{n-1} \end{bmatrix}.
\end{equation*}
It can be shown that $\xs'(s)/\|\xs'(s)\|$ is a solution to the system $\dot{q}=\Oms(s) A(s)q$, while $\frac{d}{dt}\Pperp=(\I{n}-\Pperp\Pperp\transp)A_\perp \Pperp    $, and thus $\dot{U}=\big[\Oms A\frac{\xs'{\xs'}\transp}{\|\xs'\|^2}+(\I{n}-\Pperp\Pperp\transp)A_\perp \Oms]U$. 

Consider now the following change of coordinates: $y=U(s)\xi$. We obtain
\begin{align*}
    \dot{\xi}&=U\inv\left[\Oms AU-\frac{\xs'{\xs'}\transp }{\|\xs'\|^2}A\transp U-\dot{U} \right]\xi +U\inv\Oms Bu,
\end{align*}
which it can be shown reduces to 
\begin{equation}\label{eq:xi_sys}
    \dot{\xi}=\begin{bmatrix}
    -\frac{{\xs'}\transp }{\|\xs'\|}A\frac{\xs' }{\|\xs'\|} & \0{1\times n-1} 
    \\
    \0{n-1\times 1} & \Pperp\transp A\Pperp
    \end{bmatrix}\xi
    +\begin{bmatrix}
   \0{1\times m}
    \\
    \Pperp^\dagger B
    \end{bmatrix}u.
\end{equation}
Hence, taking $\xi=[\xi_\parallel,\xi_\perp\transp]\transp$ with $\xi_\perp:=[\0{n-1\times 1},\I{n-1}]\xi$, the stability of the linearized transverse dynamics \eqref{eq:TVLsys}, and consequently the periodic orbit, is equivalent to the stability of the $(n-1)$-dimensional subsystem  \eqref{eq:TVLsubSys}. In addition, from $\dot{\xi}_\parallel=-\frac{{\xs'}\transp }{\|\xs'\|}A \frac{\xs' }{\|\xs'\|} \xi_\parallel$, it follows that we must have $\xi_\parallel=c/\|f(\xs(s))\|$ for some constant $c\in\Ri{}$. Thus \eqref{eq:NonVanSol} will be part of a non-vanishing solution of the  system \eqref{eq:TVLcompsys} regardless of the control input, while it cannot be part of a solution to the system \eqref{eq:TVLsys} due to the condition \eqref{eq:OrthogonalityCondition}.   

%% file: appendix.tex
\subsection{Proof of Lemma~\ref{lemma:spectrumCond}}\label{app:SpecCondProof}
    We begin by recalling that for a regular linear system (e.g. constant or periodic) of the form $\dot{\sigma}=C(t)\sigma$,  the sum of its characteristic exponents, denoted by $\Sigma$, is given by the formula  
    \begin{equation*}
    \Sigma=\liminf_{t\to\infty} \frac{1}{t}\int_0^t \text{Tr } C(\tau)d\tau    
    \end{equation*}
        where $\text{Tr } C$ denotes the trace of  $C$ \citep{leonov2007time}.  
    
    Also note that equivalence between the variational system of \eqref{eq:DynSYs} and the system \eqref{eq:TVLcompsys}  for $u\equiv 0$ was  demonstrated in \cite{leonov2006generalization}. Thus consider \eqref{eq:DynSYs} with  $u=K(s)\tvc=K(s)\Oms(s)\tvc$, resulting in the first approximation system
    \begin{equation*}
        \delta \dot{x}=\left(A(s)+B(s)K(s)\Oms(s)\right)\delta x. 
    \end{equation*}
    Using that $\text{Tr } CD=\text{Tr }DC$ for any $C,D\in\Ri{n\times n}$ and $\Oms^2(s)=\Oms(s)$, equivalence with \eqref{eq:TVLcompsys} follows by the same arguments as in \cite{leonov2006generalization}.
    To show equivalence between \eqref{eq:TVLcompsys} and \eqref{eq:CompSys}, we need only show  that 
    \begin{align*}
        &\liminf_{t\to\infty} \frac{1}{t}\int_0^t \text{Tr } \left( -\frac{\xs'(s(\tau)){\xs'}\transp(s(\tau))}{\|\xs'(s(\tau))\|^2}A\transp(s(\tau)) \right )d\tau, 
        \end{align*}
        which  is equal to
        \begin{align*}
                &\liminf_{t\to\infty} \frac{1}{t}\int_0^t  \left( -\frac{{\xs'}\transp(s(\tau))}{\|\xs'(s(\tau))\|}A\transp(s(\tau))\frac{\xs'(s(\tau))}{\|\xs'(s(\tau))\|} \right )d\tau,
    \end{align*}
    vanishes.   But from Sec.~\ref{sec:LimTVLD} we know this again is equivalent to 
    \begin{equation*}
        \liminf_{t\to\infty} \frac{1}{t} \ln({1/\|f(\xs(s))\|}),
    \end{equation*}
    and therefore it vanishes as desired as $f(\xs(s))$ is non-vanishing.

\subsection{Proof of Proposition~\ref{theorem:mainResult}}\label{app:ProofMR}
 Consider a coordinate change similar to the one we utilized in \ref{app:ProofOfLim}, namely $w=U(s)\chi$. This now results in
\begin{equation}\label{eq:chiSys}
    \dot{\chi}=\begin{bmatrix}
    0 &  \frac{{\xs'}\transp }{\|\xs'\|}A\transp \Pperp
    \\
    \0{n-1\times 1} & \Pperp\transp A\Pperp
    \end{bmatrix}\chi
    +\begin{bmatrix}
   \0{1\times m} 
    \\
    \Pperp^\dagger B
    \end{bmatrix}u.
\end{equation}
Hence, by defining $\chi_r:=[\0{n-1\times 1},\I{n-1}]\chi$, we  get
\begin{align*}
    \dot{\chi}_1&= \frac{\xs'(s)\transp }{\|\xs'(s)\|}A\transp(s) \Pperp(s)\chi_r 
    \\
    \dot{\chi}_r&= \Pperp\transp(s) A(s)\Pperp(s)\chi_r +\Pperp^\dagger(s) B(s) u.
\end{align*}
Therefore, unlike \eqref{eq:xi_sys}, the above subsystems are not decoupled, implying that for certain triplets $(A(\cdot),B(\cdot),\xs(\cdot))$, its origin may be asymptotically stabilized. 
Thus, taking $u=K(s) w=K(s)U(s)\chi$, we get $\dot{\chi}=A_{\chi}(s)\chi$ with 
\begin{equation*}
     A_{\chi}(s):=\begin{bmatrix}
    0 &  \frac{{\xs'}\transp }{\|\xs'\|}A\transp \Pperp
    \\[0.2cm]
    \Pperp^\dagger BK\frac{\xs'}{\|\xs'\|} & \left(\Pperp\transp A + \Pperp^\dagger BK\right)\Pperp
    \end{bmatrix}.
\end{equation*}
Also  note that with  $W(t)$ denoting the state transition  matrix of the system
\begin{equation*}
\dot{w}=\Omega(s)\left(A(s)+B(s)K(s)\right)w,    
\end{equation*}
that is $w(t)=W(t)w_0$ and $W(0)=\I{n}$, then $X(t)=U\inv(s) W(t)U(s_0)$ has the same characteristic exponents and is the state transition matrix of $\dot{\chi}=A_{\chi}(s)\chi$.  

Consider now the system \eqref{eq:TVLcompsys} with the feedback $u=K(s)\Oms(s)y$, where $\Oms(s)$ is introduced in order to satisfy the condition $\xs'(s)\transp \delta \tvc\equiv 0 $.  The dynamics of $\xi=U\inv(s) y$ are then
\begin{equation*}
    \dot{\xi}=A_{\xi}(s)\xi=A_{\chi}(s)\xi+\begin{bmatrix}
     -\frac{{\xs'}\transp }{\|\xs'\|}A\frac{\xs' }{\|\xs'\|}  &  -\frac{{\xs'}\transp }{\|\xs'\|}A\transp \Pperp
    \\[0.2cm]
    -\Pperp^\dagger BK\frac{\xs'}{\|\xs'\|} & \0{n-1}
    \end{bmatrix}\xi,
\end{equation*}
such that 
\begin{equation*}
    \xi(t)=X(t)\xi_0+X(t)\int_0^t X\inv(\tau) \tilde{A}(\tau)\xi(\tau)d\tau
\end{equation*}
where $\tilde{A}(\tau):=A_{\xi}(s(\tau))-A_{\chi}(s(\tau))$. The system is still decoupled ($\dot{\xi}_\parallel$ is independent of $\xi_\perp$, and vice versa), meaning it still has the solution $\xi_*(t)=[1/\|f(\xs(s))\|,\0{1\times (n-1)}]\transp$; hence 
\begin{equation*}
        \xi_*(0)=X\inv(t) \xi_*(t)-\int_0^t X\inv(\tau) \tilde{A}(\tau)\xi_*(\tau)d\tau.
\end{equation*}
It follows that   we can take
\begin{equation*}
\xi(0)=c\xi_*(0)+[0,\xi_\perp(0)\transp]\transp, \ \ c:=\xi_1(0)\|f(\xs(s_0))\|,
\end{equation*}
such that any solution can be written as
\begin{align*}
    \xi(t)=c\xi_*(t)&+X(t)\begin{bmatrix}0 \\ \xi_\perp(0) \end{bmatrix}
    +X(t)\int_0^t X\inv(\tau) \tilde{A}(\tau)(\xi(\tau)-c\xi_*(\tau))d\tau.
\end{align*}
Thus the system has one solution corresponding to \eqref{eq:NonVanSol}, i.e. $\xi_*(t)$, whose characteristic exponent equals  zero,  but which  is not a solution to the system \eqref{eq:TVLsys}. 
Furthermore, due to the system being decoupled, we can find $(n-1)$ additional independent solutions of the form 
\begin{equation*}
    \xi(t)=L\xi_\perp(t), \quad L\transp:=\begin{bmatrix}\0{(n-1)\times 1} & \I{n-1}\end{bmatrix},
\end{equation*}
where
\begin{align*}
    \xi_\perp(t)=L\transp X(t)\big[L\xi_\perp^0
    +\int_0^t X\inv(\tau) \tilde{A}(\tau)L\xi_\perp(\tau)d\tau\big].
\end{align*}
Therefore, as any solution of \eqref{eq:TVLsys} is of the form $\delta \tvc(t)=\Pperp(s(t))\xi_\perp(t)$, we need only find conditions ensuring the the asymptotic stability of the above solutions.
Towards this end, utilizing \eqref{eq:STMbounds} and the fact that $\|\tilde{A}L\|=\|{\xs'}\transp A\transp\Pperp\|/\|\xs'\|\le \|A\|\le \alpha $, we obtain
\begin{align*}
    \|\xi_\perp(t)\|\le C\psi(0,t)\|\xi_\perp^0\|
    +C\alpha  \int_0^t \psi(\tau,t)\|\xi_\perp\|d\tau,
\end{align*}
where $\psi(\tau,t):=\exp\left(\int_\tau^t \zeta \ ds\right)$.
Therefore, by defining $\phi(t):=\psi(t,0)\|\xi_\perp(t)\|$,  the above inequality implies
\begin{align*}
    \phi(t)\le& C\phi(0)+C\alpha  \int_0^t \phi(\tau)d\tau.
\end{align*}
This allows us to utilize Grönwall's lemma to obtain the inequality $\phi(t)\le C\phi(0)\exp\left(C\alpha t \right)$,
which  further implies 
\begin{equation*}
    \|\xi_\perp(t)\|\le C\|\xi_\perp^0\|\exp\left(\int_0^t (\zeta(s)+C\alpha)  ds\right).
\end{equation*}
Thus, by the hypothesis of the proposition, the largest characteristic exponent  therefore has a strictly negative real part and hence  the origin of \eqref{eq:LTVD}  is asymptotically (exponentially) stable.  